\documentclass[journal
,draftcls, 11pt, onecolumn
]{IEEEtran}

\usepackage[utf8]{inputenc}
\usepackage{amsmath,amssymb,amsfonts}
\usepackage{mathtools}
\usepackage{amsthm}
\usepackage{algorithmic}
\usepackage{graphicx}
\usepackage{textcomp}
\usepackage{tikz}
\usepackage{caption}
\usepackage{cuted}
\usepackage{pgfgantt}
\usepackage{pdflscape}
\usepackage{pst-plot}
\usepackage{comment} 
\usepackage{cases}

\usepackage{lineno,hyperref}
\usetikzlibrary{spy}
\usetikzlibrary{positioning,calc}
\usetikzlibrary{decorations.pathmorphing,calc,shapes,shapes.geometric,patterns}
\usetikzlibrary{shapes.multipart}
\usepackage{xfrac}
\usepackage{colortbl}
\usepackage{cancel} 
\usetikzlibrary{arrows,positioning,calc,intersections}
\usetikzlibrary{datavisualization.formats.functions}
\def\BibTeX{{\rm B\kern-.05em{\sc i\kern-.025em b}\kern-.08em
    T\kern-.1667em\lower.7ex\hbox{E}\kern-.125emX}}

\usepackage{pgfplots}
\usepgfplotslibrary{fillbetween}
\usetikzlibrary{arrows, decorations.markings}

\newtheorem{theorem}{Theorem}
\newtheorem*{theorem*}{Theorem}
\newtheorem{lemma}[theorem]{Lemma}

\newtheorem{corollary}[theorem]{Corollary}

\theoremstyle{definition}
\newtheorem{definition}[theorem]{Definition}

\theoremstyle{remark}
\newtheorem{remark}[theorem]{Remark}

\newcommand{\mbb}{\mathbb}
\newcommand{\mc}{\mathcal}

\newcommand{\seta}{\ensuremath{\mathcal{A}}}

\newcommand{\setd}{\ensuremath{\mathcal{D}}}

\newcommand{\sete}{\ensuremath{\mathcal{E}}}

\newcommand{\circlearrow}{}% just in case
\DeclareRobustCommand{\circlearrow}{%
  \mathrel{\vphantom{\rightarrow}\mathpalette\circle@arrow\relax}%
}
\newcommand{\circle@arrow}[2]{%
  \m@th
  \ooalign{%
    \hidewidth$#1\circ\mkern1mu$\hidewidth\cr
    $#1-$\cr}%
}

\begin{document}
\title{Identification over Additive Noise Channels in the Presence of Feedback} 

\author{Moritz Wiese, Wafa Labidi, Christian Deppe and Holger Boche
\thanks{M.~Wiese is with Technical University of Munich, Chair of Theoretical Information Technology, D-80333 Munich, Germany, and with CASA: Cyber Security in the Age of Large-Scale Adversaries Exzellenzcluster, Ruhr-Universit\"at Bochum, D-44780 Bochum, Germany. Email: wiese@tum.de}
\thanks{W.~Labidi is with Technical University of Munich, Chair of Theoretical Information Technology, D-80333 Munich, Germany. Email: wafa.labidi@tum.de}
\thanks{H.~Boche is with Technical University of Munich, Chair of Theoretical Information Technology, D-80333 Munich, Germany, the BMBF Research
Hub 6G-life, and with CASA: Cyber Security in the Age of Large-Scale Adversaries Exzellenzcluster, Ruhr-Universit\"at Bochum, D-44780 Bochum, Germany. Email: boche@tum.de}
\thanks{C.~Deppe is with Technical University of Munich, Institute for Communications Engineering, D-80333 Munich, Germany. Email: christian.deppe@tum.de}
%\thanks{M.~Wiese was supported by the German Research Foundation (DFG) within Germany’s Excellence Strategy EXC-2092 CASA-390781972 and within the Gottfried Wilhelm Leibniz Prize under Grant BO 1734/20-1. W.~Labidi was supported by the German Federal Ministry of Education and Research (BMBF) within the national initiative for “Post Shannon Communication (NewCom)” through the project “Basics, simulation and demonstration for new communication models” under Grant 16KIS1003K and within the national initiative for ``Molecular Communications'' (MAMOKO) under Grant 16KIS0914. C.~Deppe was supported by BMBF within NewCom through the project “Coding theory and coding methods for new communication models” under Grant 16KIS1005 and within the national initiative on 6G Communication Systems through the research hub 6G-life under Grant 16KISK002. H.~Boche was supported in part by BMBF within 6G-life under Grant 16KISK002, within NewCom under Grant 16KIS1003K, by the Bavarian Ministry of Economic Affairs, Regional Development and Energy as part of the project 6G Future Lab Bavaria, and by DFG within the Gottfried Wilhelm Leibniz Prize under Grant BO 1734/20-1.}
\thanks{This work has been presented in part at the virtual IEEE International Symposium on Information Theory (ISIT) 2021.}
}

\maketitle

\begin{abstract}
We analyze deterministic message identification via channels with non-discrete additive white noise and with a noiseless feedback link under both average power and peak power constraints. The identification task is part of Post Shannon Theory. The consideration of communication systems beyond Shannon's approach is useful in order to increase the efficiency of information transmission for certain applications. We propose a coding scheme that first generates infinite common randomness between the sender and the receiver. If the channel has a positive message transmission feedback capacity, for given error thresholds and sufficiently large blocklength this common randomness is then used to construct arbitrarily large deterministic identification codes. In particular, the deterministic identification feedback capacity is infinite regardless of the scaling (exponential, doubly exponential, etc.) chosen for the capacity definition. Clearly, if randomized encoding is allowed in addition to the use of feedback, these results continue to hold.
\end{abstract}

\begin{IEEEkeywords}
Identification theory, feedback, common randomness, additive noise channels
\end{IEEEkeywords}

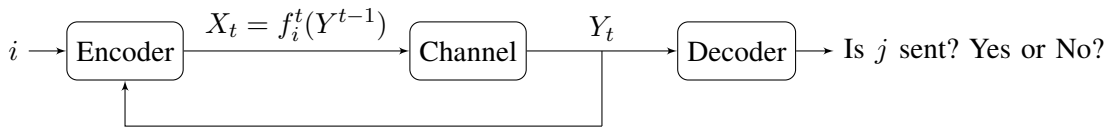
\begin{figure*}
\centering
\tikzstyle{block} = [draw, rectangle, rounded corners,
minimum height=2em, minimum width=2.5em]
\tikzstyle{input} = [coordinate]
\tikzstyle{sum} = [draw, circle,inner sep=0pt, minimum size=2mm,  thick]
%\scalebox{.9}{
\tikzstyle{arrow}=[draw,->] %{Latex[length=3mm]},
\begin{tikzpicture}[auto, node distance=2cm,>=latex']
\node[] (M) {$i$};
\node[block,right=.5cm of M] (enc) {Encoder};
\node[block, right=3cm of enc] (channel) {Channel};
\node[block, right=2cm of channel] (dec) {Decoder};
\node[right=.5cm of dec] (Mhat) {Is $j$ sent? Yes or No?};
%\node[above=1cm of channel] (noise) {$Z_t$};
\node[input,right=1cm of channel] (t1) {};
\node[input,below=1cm of t1] (t2) {};
\draw[->] (M) -- (enc);
\draw[->] (enc) --node[above]{ $X_t=f_i^t(Y^{t-1})$} (channel);
%\draw[->] (noise) -- (channel);
\draw[->] (channel) --node[above]{$Y_t$} (dec);
\draw[->] (dec) -- (Mhat);
\draw[-] (t1) -- (t2);
\draw[->] (t2) -| (enc);
\end{tikzpicture}
%}
\caption{Discrete-time memoryless channel with noiseless feedback}
\label{fig:new_channelModel}
\end{figure*}

\section{Introduction}
New applications in modern communications demand robust and ultra-reliable low latency information exchange such as machine-to-machine and human-to-machine communications \cite{CompoundChannel}, the tactile internet \cite{tactileInternet}, digital watermarking \cite{MOULINwatermarking,AhlswedeWatermarking,SteinbergWatermarking}, health care, industry 4.0, etc. Novel communication tasks like molecular communication \cite{molecularCommunication} also pose new challenges. For many of these applications, the identification approach suggested by Ahlswede and Dueck \cite{AhlDueck} in 1989 is much more suitable than the classical message transmission scheme proposed by Shannon \cite{Shannon}. For this reason, intensive research has been started recently \cite{FettwBoche_6G}, \cite{FitzBoch_Netwk} in order to investigate the potential of identification and related communication tasks for application in future communication systems and to find the corresponding capacity characterizations as well as efficient coding schemes for practically relevant communication scenarios.

In the classical message transmission scheme, the encoder transmits a message over a channel. At the receiver side, the decoder aims to estimate this message based on the channel observation. In contrast, in the identification scheme, the sender and the receiver are given an \textit{identity} each. The sender encodes his identity, and the receiver needs to check whether or not the identities coincide.

The identification problem can be regarded as the task of performing many hypothesis tests simultaneously. The starting point of the theory of identification which immediately sparked great interest in this new paradigm was the result of \cite{AhlDueck} for Discrete Memoryless Channels (DMCs) which states that the size of identification codes grows doubly exponentially fast with the blocklength, if randomized encoding is allowed. This is dramatically different from the classical message transmission, where the number of messages that can be reliably communicated over the channel is exponential in the blocklength. However, in this result, randomized encoding is essential in order to achieve the double exponential growth. In the deterministic setup, the number of messages that can be identified over a DMC only scales exponentially with the blocklength \cite{deterministicDMC,deterministicFading,IDwithoutRandom}. However, in the case of deterministic encoding, the rate is still larger than the message transmission rate in the exponential scale. Apart from these gains, other communication scenarios such as correlation-assisted identification \cite{correlation}, secure correlation-assisted identification \cite{globecom} as well as identification in the presence of feedback \cite{IDFeedback,FeedbackBehavior} show that the identification problem produces completely new effects compared to Shannon's message transmission problem.

The availability of a feedback channel has been shown not to increase the Shannon capacity of a DMC, even if the feedback is noiseless and has unlimited capacity \cite{Shannon56}. However, it can help greatly in reducing the complexity of encoding or decoding \cite{ConstructiveProofAhlswede}. Furthermore, it has been proved in \cite{MACtransmissionFeedback,DUECK19801,KramerFeedback} that feedback increases the capacities of discrete memoryless multiple-access channels as well as discrete memoryless broadcast channels. The authors of \cite{Ahlswede2006} pointed out that the noiseless feedback can be used to generate a secret key shared only between the transmitter and the legitimate receiver. 

The combination of identification with noiseless feedback was studied by Ahlswede and Dueck \cite{IDFeedback} when the channel is a DMC. They showed that, even in the case of deterministic encoding, feedback allows the number of identities to grow doubly exponentially in the blocklength. The feedback allows us to set up a common randomness experiment shared by the sender and the receiver which can be used to construct an efficient identification code. The amount of correlated randomness determines the doubly-exponential growth rate of the identification capacity and is given by the maximal entropy of any output distribution which can be generated over the channel. If in addition, it is allowed to use randomized encoding, the capacity grows even larger, although still on the doubly-exponential scale. These results are special cases of the fact that the identification capacity on the doubly-exponential scale of a DMC coincides with the capacity of common randomness \cite{Ahlswede2006}.

Other work on identification in the presence of feedback has focused on channels with finite input and output alphabets. Identification via discrete arbitrarily varying channels (AVC) with noiseless feedback was investigated in \cite{Ahlswede2000}. Identification over discrete multi-way channels with complete feedback was presented in \cite{ID_Multiway_Feedback}. In \cite{GeneralTheory}, Ahlswede established a unified theory of identification via channels with finite input and output alphabets in the presence of noisy feedback. Secure identification over the discrete memoryless wiretap channel in the presence of secure feedback was studied in \cite{Ahlswede2006}. 

Only a few studies \cite{HanBook,Burnashev,deterministicFading,icassp_paper,globecom} have explored identification for continuous alphabets, although such channels are highly relevant for applications such as those mentioned in \cite{FettwBoche_6G}. We are concerned with channels with \emph{non-discrete additive white noise}. The most prominent example of such a channel is the channel with additive white Gaussian noise (AWGN), which is practically relevant in wired and wireless communications, satellite and deep space communication links, etc. Unusual phenomena are encountered when we extend the identification problem from the DMC case to the case of continuous alphabets. For instance, the maximum size of deterministic identification codes without feedback for the AWGN channel scales as $n^{Rn}$ for some positive $R$ as the blocklength $n$ ends to infinity, which is neither singly- nor doubly-exponential \cite{deterministicFading}. But like in the DMC case, the size of identification codes without feedback scales doubly exponentially fast in the blocklength if randomized encoding is permitted \cite{Burnashev,MasterThesis}.

Although identification with feedback currently is an active research area, no results have yet been established for continuous alphabets in the presence of feedback. We determine the identification capacity of channels with non-discrete additive white noise in the presence of noiseless feedback for the case of deterministic encoding. See Fig. \ref{fig:new_channelModel} for an illustration of the problem setup. While the average power constraint provides analytical tractability, real-life systems are limited in their peak power. In our case, we are able to find the identification feedback capacity not only subject to an average power constraint, but also that subject to peak power constraint. 

In fact, we find even more than a capacity result. If the channel with non-discrete additive white noise has a positive message transmission capacity under the given power constraint, for given error thresholds and sufficiently large blocklength we construct arbitrarily large deterministic identification codes. In other words, only the error thresholds determine the minimum necessary blocklength, and once this requirement is met, the size of the identification code can be chosen independently of the blocklength. Consequently, no matter with respect to which scale (exponential, doubly exponential, \ldots) one defines the deterministic identification feedback capacity, the capacity will always be positive. We also formalize this "capacity-theoretic" point of view. Clearly, if randomized encoding is allowed in addition to the use of feedback, these results continue to hold. For this reason, we will not formally introduce or discuss the case of randomized encoding together with perfect feedback. Our result is yet another example where the identification task shows a completely different behavior than Shannon's message transmission task.

The reason for our surprising result is that feedback allows the sender and the receiver to establish a shared random experiment on an arbitrarily large finite set. This permits us to make the error probabilities of the second kind arbitrarily small, which describe the probability that the receiver wrongly decides that his identity is the same as the sender's. The decision sets corresponding to different identities overlap, and using the arbitrarily large common randomness, one can devise arbitrarily strong "challenges" which help to distinguish these decision sets.

% We show that all positive identification rates are achievable in the double exponential scaling, i.e., the corresponding capacity is infinite. This raises the question whether an appropriate scaling different from the double exponential scaling exists, such that the identification capacity is finite, as is shown in \cite{deterministicFading} for deterministic identification over fading channels. 
% Surprisingly, we prove that for positive noise variance and a \emph{non-discrete additive white noise}, we can build for given $\lambda_1,\lambda_2\in (0,1)$ a deterministic identification feedback code $(m,N,\lambda_1,\lambda_2)$ of sufficiently large blocklength $n$ with arbitrarily large $N$ identities. Our proposed coding scheme allows the generation of infinite common randomness between the sender and the receiver.

In Section \ref{sec:preliminaries}, we introduce our system model, present the main results of the paper and provide a comparative discussion of related work.
In Section \ref{sec: ID_GaussianChannel_Feedback}, we provide a coding scheme that generates infinite common randomness between the sender and the receiver in our system model and prove the main results of the paper. Section \ref{sec: conclusions} contains concluding remarks and proposes potential future research in this field.

\section{System Model and Main Result}\label{sec:preliminaries}

In this section, we introduce the notation that will be used throughout the paper. We introduce our system model and present the main result of the paper. We conclude with a discussion of the result and a comparison with related and similar results.

\subsection{Notation}

The letter $\mbb R$ denotes the set of real numbers. By $\log$, we mean the logarithm with base 2, whereas $\ln$ is the natural logarithm.

\subsection{System Model}

\subsubsection{Channels with additive noise}

The channels over which we want to perform identification are discrete-time memoryless channels with additive white noise, with real inputs and outputs. An input $x\in\mbb R$ results in an output variable $Y$ of the form
\[
    Y=x+Z,
\]
where $Z$ is the additive noise whose distribution $P$ is independent of $x$. Our strategy achieving an infinite identification capacity does not apply to arbitrary noise distributions. The admissible types of noise can be described using the Lebesgue decomposition of probability measures. In order to formulate the Lebesgue decomposition, we need to define three properties of probability measures. 

\begin{definition}
Let $P$ be a probability distribution on $\mbb R$ with cumulative density function (cdf) $F$.
\begin{enumerate}
    \item $P$ and $F$ are called \emph{absolutely continuous} if there exists a measurable function $f(x)$ such that
\begin{equation*}
    P[\mc A]=\int_{\seta} f(x) \ dx
\end{equation*}
for all measurable $\mc A$. The function $f$ is called the \textit{probability density function (pdf)} of the random variable $X$.

    \item $P$ and $F$ are called \emph{discrete} if there exists a measurable set $\mc A\subset\mbb R$ such that $P(\mc A)=1$ and such that $P[\{x\}]>0$ for all $x\in\mc A$.
    
    \item $P$ and $F$ are called \emph{singular continuous} if $F$ is continuous and $P$ is singular with respect to the Lebesgue measure (i.e., there exists a measurable $\mc A\subset\mbb R$ such that $P[\mc A]=1$ and $\int_{\mc A}dx=0$).
\end{enumerate}
\end{definition}

Note that the three properties defined above are mutually exclusive. A given probability distribution does not need to have any of these properties. However, by Lebesgue decomposition, any probability distribution can be decomposed into an absolutely continuous, a discrete and a singular continuous part.

\begin{lemma}[Lebesgue decomposition. See, e.g., \cite{K_Formin}]
    Let $F$ be a cumulative distribution function. Then there exists a triple of cumulative distribution functions $(D, A, S)$ such that $D$ is discrete, $A$ is absolutely continuous and $S$ is singular continuous, and nonnegative numbers $p_D,p_A,p_S$ satisfying $p_D+p_A+p_S=1$, such that
    \[
        F=p_DD+p_AA+p_SS.
    \]
\end{lemma}

We can now define the relevant properties of the distribution $P$ of the channel noise.

\begin{definition}
    Let $P$ be a probability measure with cdf $F$. 
    \begin{enumerate}
        \item If $F$ is continuous, which means that there exist an absolutely continuous cdf A and a singular cdf S such that $F=p_AA+p_SS$ for some nonnegative $p_A,p_S$ with $p_A+p_S=1$, then also $P$ is called \textit{continuous}.
        \item If $F$ is not discrete, then $F$ and $P$ are called \textit{non-discrete}.
    \end{enumerate}
\end{definition}

Our main result on the achievability of infinite identification capacity with perfect feedback holds for channels with additive noise whose distribution is non-discrete. The reason why we need non-discrete noise is that for such noise, the discrete part can be removed by suitable conditioning. What remains is distributed according to a continuous distribution and can therefore be transformed into a uniform distribution on an arbitrarily large finite set. This is the key reason for the achievability of infinite identification feedback capacity.

We can now define the type of channels considered in this work.

\begin{definition}
\label{def_channelModel}
 Let $P$ be a non-discrete probability distribution. The \textit{discrete-time memoryless channel $W_P$ with non-discrete additive white noise} has real inputs and outputs. An input $(x_1,\ldots,x_n)\in\mbb R^n$ of length $n$ generates the output random vector
 \begin{equation}\label{eq:n_channel_uses}
     (Y_1,\ldots,Y_n)=(x_1,\ldots,x_n)+(Z_1,\ldots,Z_n)
 \end{equation}
 of the same length, where $Z_1,\ldots,Z_n$ are i.i.d.\ copies of a noise random variable $Z$ with probability distribution $P$. The probability distribution of the random vector \eqref{eq:n_channel_uses} will be denoted by $W_P^n(\cdot\vert x_1,\ldots,x_n)$.
\end{definition}

For later application, we now have a closer look at the output distribution $W_P(\cdot\vert x)$ generated by the channel input $x\in\mbb R$. First of all, by \eqref{eq:n_channel_uses}, it is a shifted version of the noise distribution $P$. More precisely, for any measurable $\mc A\subset\mbb R$,
\[
    W_P(\mc A\vert x)=P(\mc A-x),
\]
where $\mc A-x=\{y:y+x\in\mc A\}$. Moreover, it is well-known that any probability distribution on the reals is associated with an integral operator. For instance, for any measurable $\mc A\subset\mbb R$, we can write
\[
    W_P(\mc A\vert x)=\int_{\mc A}W_P(dy\vert x).
\]
More generally, the expectation of the measurable function $g$ with respect to the probability measure $W_P(\cdot\vert x)$ is given by
\begin{equation}\label{eq:W_P_exp}
    \int g(y) W_P(dy\vert x).
\end{equation}

\begin{remark}\label{rem:abscont}
Note that these integrals can only be written as integrals with respect to the Lebesgue measure if $W_P(\cdot\vert x)$ is absolutely continuous, which is the case if and only if $P$ is absolutely continuous. If $P$ is absolutely continuous with the pdf $f$, then 
\[
    W_P(\mc A\vert x)=\int_{\mc A} f(y-x)\,dy.
\]
The standard example of an absolutely continuous probability measure is the Gaussian distribution, and the corresponding channel is the AWGN channel.
\end{remark}

\subsubsection{Identification feedback codes}

We assume that perfect feedback is available to the sender. This means that after the transmission of every symbol, the sender obtains a perfect copy of the received symbol and can use this knowledge when choosing the next symbol to send. The availability of perfect feedback is reflected by the form of the encoding functions which will be used in identification feedback codes. 

\begin{definition}
 A \textit{feedback encoding function of length $n$} is a family $f=(f^1,\ldots,f^n)$ of $n$ real-valued functions satisfying
\begin{align*}
    f^1 & \in\mbb R,\\ 
    f^2    & \colon \mbb R\to\mbb R, \\
    \vdots & \\
    f^n & \colon \mbb R^{n-1} \longrightarrow \mbb R.
\end{align*}
For any $\Gamma\geq 0$, the set of feedback encoding functions satisfying the average power constraint
\begin{align*}
    &\sum_{t=1}^{n} (f^t(y_1,\ldots,y_{t-1}))^2 \leq n \Gamma
    \qquad\text{for all }(y_1,\ldots,y_{n-1})\in\mbb R^{n-1}
\end{align*}
is denoted by $\overline{\mc F}_{n,\Gamma}$. The set of feedback encoding functions satisfying the peak power constraint
\begin{align*}
    &|f^t(y_1,\ldots,y_{t-1})| \leq  \Gamma\qquad  \text{for all } t\in \{1,\ldots,n\}
    \text{ and all }(y_1,\ldots,y_{t-1})\in\mbb R^{t-1}
\end{align*}
is denoted by $\mc F_{n,\Gamma}$.
\end{definition}

In an identification feedback code with $N$ identities, any identity $i$ will be associated with a feedback encoding function $f_i$. If the blocklength is $n$, then the sender starts by sending the real number $f_i^1$. From the perfect feedback, it obtains the receiver's first channel output $y_1$ and sends $f_i^2(y_1)$ in the second channel use. This continues until finally, at the $n$-th channel use, the sender knows all previous channel outputs $y_1,\ldots,y_{n-1}$ and sends the symbol $f_i^n(y_1,\ldots,y_{n-1})$.

We want to formalize this by defining deterministic identification feedback codes. Before we can do this, we need to describe the output probability $W_P^n(\cdot\vert f)$ generated by any feedback encoding function $f$ of length $n$. This is more complicated than for the transmission of a simple symbol sequence $(x_1,\ldots,x_n)$ because the feedback has to be taken into account. 

Recall that any probability distribution on $\mbb R^n$ is characterized by the values it assumes on product sets \cite[pp.~144f.]{Shiryaev}. Thus it is sufficient to describe $W_P^n(\mc A_1\times\cdots\times\mc A_n\vert f)$ for measurable sets $\mc A_1,\ldots,\mc A_n\subset\mbb R$. Proceeding by induction over $n$, it is not hard to see that
\begin{align*}
    &W_P^n(\mc A_1\times\dots\times\mc A_n\vert f)\\
    &=\int_{\mc A_1}\int_{\mc A_2}\cdots\int_{\mc A_n}W_P\bigl(dy_n\vert f^{n-1}(y_1,\ldots,y_{n-1})\bigr)\cdots W_P\bigl(dy_2\vert f^2(y_1)\bigr)W_P\bigl(dy_1\vert f^1\bigr)
\end{align*}
(recall \eqref{eq:W_P_exp}). This characterization now permits us to also write $W_P^n(\mc D\vert f)$ for arbitrary measurable sets $\mc D\subset\mbb R^n$.

\begin{definition}\label{Def:strategy_f}
Let $W_P$ be a channel with non-discrete additive white noise. Let $n,N$ be positive integers and $\lambda_1,\lambda_2$ nonnegative reals satisfying $\lambda_1+\lambda_2<1$. 

\begin{enumerate}
    \item An $(n,N,\lambda_1,\lambda_2)$ \textit{deterministic identification feedback code for $W_P$ with average power constraint $\Gamma\geq 0$} is a family of pairs $\left \{(f_i,\setd_i):i=1,\ldots,N \right\}$ with
\begin{align*} 
    &f_i=(f_i^1,f_i^2\ldots,f_i^n) \in \overline{\mc F}_{n,\Gamma},\quad\mathcal D_{i} \subset \mbb R^{n}\qquad\text{for all }i\in \{1,\ldots,N\} 
\end{align*}
satisfying
\begin{align}
\mu_1^{(i)} \triangleq W_P^n(\setd_i^c|f_i) & \leq \lambda_1 \quad  \text{for all } i, \label{error1}\\
\mu_2^{(i,j)} \triangleq W_P^n(\setd_j|f_i) & \leq \lambda_2 \quad \text{for all }i\neq j. \label{error2}
\end{align}
    
    \item An $(n,N,\lambda_1,\lambda_2)$ \textit{deterministic identification feedback code for $W_P$ with peak power constraint $\Gamma\geq 0$} is defined in an analogous manner but with all feedback encoding functions $f_i$ contained in $\mc F_{n,\Gamma}$ instead of $\overline{\mc F}_{n,\Gamma}$.
\end{enumerate}
\end{definition}

Before we discuss the meaning of \eqref{error1} and \eqref{error2}, let us first recall the definition of a traditional message transmission feedback code. 

\begin{definition}\label{Def:transmission_feedback}
Let $W_P$ be a channel with non-discrete additive white noise. Let $n,M$ be positive integers and $\lambda\geq 0$. An $(n,M,\lambda)$ \textit{deterministic message transmission feedback code for $W_P$ with average power constraint $\Gamma\geq 0$} is a family of pairs $\left \{(f_i,\setd_i):i=1,\ldots,M \right\}$ with
\begin{align*} 
    &f_i=(f_i^1,f_i^2\ldots,f_i^n) \in \overline{\mc F}_{n,\Gamma},\quad\mathcal D_{i} \subset \mbb R^{n}\qquad\text{ for all }i\in \{1,\ldots,M\},\\
    &\mc D_i\cap\mc D_j=\varnothing\quad \text{ for }i\neq j 
\end{align*}
satisfying
\begin{align}\label{eq:transmission_error}
    W_P^n(\setd_i^c|f_i) & \leq \lambda \quad  \text{for all } i. 
\end{align}
If every $f_i$ is contained in $\mc F_{n,\Gamma}$, then the code is an $(n,M,\lambda)$ \textit{deterministic message transmission feedback code for $W_P$ with peak power constraint $\Gamma$}.
\end{definition}

Definition \ref{Def:transmission_feedback} is well-known. The main difference between a message transmission feedback code and an identification feedback code is that the former also requires the decoding sets to be disjoint. The absence of this requirement introduces new kinds of error events for identification codes. The message transmission error probability \eqref{eq:transmission_error} is analogous to the error probability $\mu_1^{(i)}$ defined in \eqref{error1}. We call $\mu_1^{(i)}$ a \textit{probability of error of the first kind}. 

For a pair $(i,j)$ of distinct identities, the $\mu_2^{(i,j)}$ defined in \eqref{error2} is called a \textit{probability of error of the second kind}, and gives the probability that the decoder decides for identity $j$ when in fact identity $i$ was sent. In traditional message transmission, an error of the first kind automatically results in an error of the second kind, and vice versa. This is due to the requirement of disjoint decoding sets which is present in message transmission codes. Since the decoding sets of an identification code do not need to be disjoint, an error of the second kind can result from the overlapping of the decoding sets and does not need to be accompanied by an error of the first kind. 

The last preparation for the statement of our main results is the definition of message transmission feedback capacity. We will see why we do not need to define an identification feedback capacity, or at least, why it does not make any sense to define a \textit{single} such capacity, when we state the results. 

\begin{definition}
 Let $W_P$ be a channel with non-discrete additive white noise and $\Gamma\geq 0$. A number $R\geq 0$ is called an \textit{achievable message transmission feedback rate with (peak/average) power constraint $\Gamma$} if for every $\lambda,\delta>0$ and sufficiently large blocklength $n$ there exists an $(n,M,\lambda)$ message transmission feedback code for $W_P$ with (peak/average) power constraint $\Gamma$ satisfying 
 \[
    \frac{\log M}{n}\geq R-\delta.
 \]
 The maximal achievable message transmission feedback rate with power constraint $\Gamma$ is called the \textit{message transmission feedback capacity of $W_P$ with (peak/average) power constraint $\Gamma$}.
\end{definition}

\subsection{Main Results}

The main result of the paper is the following theorem. It holds for both peak and average power constraints; we will see that it is irrelevant for the proof strategy which type of constraints is used. This is good news, since peak power constraints are found frequently in practical applications, whereas average power constraints usually are much more amenable to analysis.

\begin{theorem}\label{theorem:new_main_theorem} 
Let $\lambda\in(0,\frac{1}{2})$, $\Gamma>0$ and $W_P$ a channel with non-discrete additive white noise.

\begin{enumerate}
    \item If $W_P$ has positive message transmission feedback capacity with average power constraint $\Gamma$, then there exists a blocklength $n_s$ such that for every positive integer $N$ and every $n\geq n_s$ there exists a deterministic $(n,N,\lambda_1,\lambda_2)$ identification feedback code for $W_{P}$ with average power constraint $\Gamma$ and with $\lambda_1,\lambda_2\leq\lambda$.
    \item If $W_P$ has positive message transmission feedback capacity with peak power constraint $\Gamma$, then there exists a blocklength $n_s$ such that for every positive integer $N$ and every $n\geq n_s$ there exists a deterministic $(n,N,\lambda_1,\lambda_2)$ identification feedback code for $W_{P}$ with peak power constraint $\Gamma$ and with $\lambda_1,\lambda_2\leq\lambda$.
\end{enumerate}
\end{theorem}

Note that the blocklength $n_s$ in the statement of the theorem only depends on the error probabilities $\lambda_1,\lambda_2$, but not on the number of identities $N$. Moreover, and this is the most surprising aspect of the theorem, if $W_P$ has positive message transmission feedback capacity, then as soon as the blocklength is large enough subject only to the required error bounds, the number of identities can be chosen arbitrarily large. 

Clearly, this result means that the "identification feedback capacity" of a channel $W_P$ with positive message transmission feedback capacity is infinite. This holds both for the case when "capacity" measures the logarithm of the largest number of possible identities per channel use and when it measures the double logarithm of the largest number of identities per channel use. In order to put Theorem \ref{theorem:new_main_theorem} into the context of other capacity results, we define the term "capacity" with respect to an arbitrary rate function.

 \begin{definition} \label{Def:phi-rate} 
 Let $\varphi$ be an arbitrary continuous strictly monotonically increasing \emph{rate function} $\varphi \colon \mathbb{R}^+ \to \mathbb{R}^+$ with $\lim_{x \to \infty} \varphi(x)= + \infty$ and let $\Gamma\geq 0$.
 \begin{enumerate}
 \item The identification rate $R$ for the channel $W_{P}$ is called \textit{achievable with respect to (w.r.t.) the rate function $\varphi$ and with (average/peak) power constraint $\Gamma$} if for every $\lambda \in (0,\frac{1}{2})$ there exists an $n_s(\lambda)$ such that for all $n\geq n_s(\lambda)$ there exists an $(n,\varphi(n R),\lambda,\lambda)$ deterministic identification feedback code for $W_{P}$ with (average/peak) power constraint $\Gamma$.
  \item The \textit{deterministic feedback identification capacity $\overline C^\varphi_{IDf}(W_P,\Gamma)$ of $W_{P}$ w.r.t. the rate function $\varphi$ and with average power constraint $\Gamma$} is the supremum of all achievable rates w.r.t.\ $\varphi$ and with average power constraint $\Gamma$.
  \item The \textit{deterministic feedback identification capacity $C^\varphi_{IDf}(W_P,\Gamma)$ of $W_{P}$ w.r.t. the rate function $\varphi$ and with peak power constraint $\Gamma$} is the supremum of all achievable rates w.r.t.\ $\varphi$ and with peak power constraint $\Gamma$.
  \end{enumerate}
 \end{definition}
 
The rate function used in the traditional definition of message transmission capacity is $\varphi_1(x)=2^x$. The double-exponential increase in blocklength of the number of identities which is achievable in identification for discrete memoryless channels with randomized encoding is reflected in the use of the rate function $\varphi_3(x)=2^{2^x}$ in the corresponding capacity definition. More details will be discussed below, where we also give the example of a rate function $\varphi_2$ which ranges between $\varphi_1$ and $\varphi_3$.
 
Theorem \ref{theorem:new_main_theorem} now has the following corollary.
 
\begin{corollary}
Let $\varphi$ be an arbitrary rate function as defined in Definition \ref{Def:phi-rate} and let $\Gamma\geq 0$. If the channel $W_P$ with non-discrete additive white noise has positive message transmission feedback capacity with average power constraint $\Gamma$, then $C^{\varphi}_{IDf}(W_P,\Gamma)=+\infty$. If $W_P$ has positive message transmission feedback capacity with peak power constraint $\Gamma$, then $C_{IDf}(W_P,\Gamma)=+\infty$.
\label{corrolary}
\end{corollary}

\subsection{Comparison and Discussion}

In this discussion section, we will also use the rate functions of Definition \ref{Def:phi-rate} to describe the growth rates of identification codes for other types of identification, for instance identification without feedback or with randomized encoding. For any given type of identification, one can define rates achievable with respect to a given rate function $\varphi$ exactly analogously as we did for identification with feedback in Definition \ref{Def:phi-rate}.

\subsubsection{Gaussian channels}

Let us consider the case where $W_P$ is a Gaussian channel, i.e., $P$ is a normal distribution with mean 0 and variance $\sigma^2>0$. As noted in Remark \ref{rem:abscont}, this is an absolutely continuous probability distribution and the message transmission capacities of $W_P$ both for average and peak power constraint are positive for all positive $\Gamma$. 

In the case of deterministic encoding without feedback, it turns out that the identification capacity of the Gaussian channel w.r.t. $\varphi_1(x)=2^x$ is infinite and zero w.r.t. $\varphi_3=2^{2^x}$. It was shown in \cite{deterministicFading} that the identification capacity w.r.t.\ the rather unusual intermediate rate function $\varphi_2(x)=x^x=2^{x\log x}$ is positive and finite. Hence in this problem setup, the maximal number of identities grows superexponentially in the blocklength, but slower than doubly-exponentially. By Theorem \ref{theorem:new_main_theorem}, this behavior changes radically when we add perfect feedback. Thus feedback here has a dramatically different effect than in the message transmission case, where the addition of feedback leaves the capacity unchanged.

Generally, one can observe that identification is more sensitive with respect to different problem setups than message transmission. It has been proved that for identification without feedback, but with randomized encoding over the Gaussian channel \cite{Burnashev,MasterThesis}, the size of identification codes scales doubly exponentially fast in the blocklength. That means that the corresponding capacity w.r.t.\ the rate function $\varphi_3(x)=2^{2^x}$ is positive and finite.

\subsubsection{Discrete memoryless channels}

DMCs are somewhat less surprising when it comes to the effect of feedback on the identification capacity. Still, feedback does have an effect, in contrast to Shannon's result \cite{Shannon56} that the message transmission capacity of a DMC does not change whether or not feedback is available.

In the case of DMCs, the identification feedback capacity with deterministic encoders is positive but finite w.r.t.\ the doubly exponential rate function $\varphi_3$ if the DMC has positive capacity, but is not noiseless. In fact, if it is positive, this capacity is given by the maximal output entropy which can be produced by any input to the DMC \cite{IDFeedback}. This result indicates the importance of channel noise for the generation of common randomness shared by sender and receiver, which is necessary to achieve such a large rate in the absence of randomized encoding. (Without feedback and randomized encoding, only the capacity w.r.t.\ the singly exponential rate function $\varphi_1$ is positive and finite \cite{IDwithoutRandom, deterministicDMC}.) 

In fact, our proof strategy for Theorem \ref{theorem:new_main_theorem} follows \cite{IDFeedback}. The reason why we obtain an infinite identification feedback capacity is that the sender and the receiver can generate an infinite amount of common randomness. That the identification feedback capacity is not infinite in the discrete case is due to the fact that it is impossible to generate uniform random experiments on arbitrarily large finite sets from a finite number of channel outputs of a DMC. If, in addition, randomized encoding is allowed, a further increase of the identification capacity can be achieved for DMCs, but without changing the rate function \cite{IDFeedback}.

%%%%%%%%%%%%%%%%%%%%%%%%%%%%%%%%%%%%%%
 
\section{Proof of the Main Result}\label{sec: ID_GaussianChannel_Feedback}

In this section, we provide a proof of Theorem \ref{theorem:new_main_theorem} and Corollary \ref{corrolary}. We start by giving a short overview of our proof strategy. We associate every identity $i$ with a "coloring function" $k_i:\{1,\ldots,L\}\to\{1,\ldots,M\}$. Assume that the sender and the receiver have access to the outcome of a random experiment which uniformly at random chooses an $l\in\{1,\ldots,L\}$. If the sender wants to send identity $i$, it computes $k_i(l)$ and sends this "color" to the receiver using a message transmission feedback code. The receiver interested in identity $j$ decodes the channel output to a value $\hat m\in\{1,\ldots,M\}$. It knows the value $l$, and so it can test whether $k_j(l)=\hat m$. If this is the case, then it determines that identity $j$ was sent. Otherwise, it determines that $j$ was not sent. The number of identities which can be distinguished in this way with small error probabilities grows to infinity with $L$ if we can find message transmission feedback codes of positive rate with arbitrarily small error probability.

In order to make this strategy work, we need to find a way to implement the joint random experiment. This is where feedback helps us: We can just use the channel noise as a natural source of randomness. By feedback, the channel noise is also known to the sender. In the first subsection below, we will see how to generate from non-discrete channel noise a uniform distribution on $\{1,\ldots,L\}$ for arbitrarily large $L$. In the second subsection, we will formally define our identification feedback codes and analyze the error probabilities as required for the proof of Theorem \ref{theorem:new_main_theorem}. The third subsection contains the proof of Corollary \ref{corrolary}.

\subsection{Common Randomness Generation}\label{sec:CR}

By sending the symbol 0, the sender generates a random variable $Z$ observed by the receiver and distributed according to the noise distribution $P$. Through feedback, the sender, just like the receiver, knows the precise noise realization. If $P$ is a continuous distribution, it is possible to partition $\mbb R$ into $L$ subsets of equal probability. In this way, the sender and the receiver can generate a uniform random experiment on $\{1,\ldots,L\}$. 

If $P$ is non-discrete, but not continuous, this is still possible, but more work has to be done. The i.i.d.\ noise has to be realized several times, and if at least one of the noise samples is in the "continuous part" of the noise distribution, a uniform distribution on $\{1,\ldots,L\}$ can be generated as in the continuous case. However, there remains a positive probability that this process fails, which has to be taken into account later when error probabilities are analyzed.

Recall that Lebesgue decomposition permits us to represent the cdf $F$ of $P$ as
\begin{equation*}
    F(z)=p_DD+p_AA+p_SS,
\end{equation*}
where $D$ is a discrete cdf, $A$ an absolutely continuous cdf and $S$ a singular continuous cdf and the nonnegative numbers $p_D,p_A,p_S$ add up to 1. Define $\bar{\sete}$ as the set of jump points of $D$, i.e.,
\begin{equation*}
    \bar{\sete}=\{ z: P[\{z\}]>0 \},
\end{equation*}
and set $\sete=\mathbb{R} \setminus \bar{\sete}$. Since $F$ is non-discrete by assumption, the "continuous part" $p_AA+p_SS$ of $F$ does not vanish, hence $P[\mc E]=p_A+p_S>0$. Therefore the function
\[
    F'=\frac{p_AA+p_SS}{P[\mc E]}
\]
is the cdf of a continuous probability distribution $P'$. In fact, we have the following relation between $P$ and $P'$.

\begin{lemma}
    $P'=P[\cdot\vert\mc E]$.
\end{lemma}

\begin{proof}
    It is sufficient to check that $P[\cdot\vert\mc E]$ and $P'$ have the same cdf. As long as we do not know that they are equal, denote the cdf of $P[\cdot\vert\mc E]$ by $\tilde F$. Also, denote the probability measures corresponding to $D, A$ and $S$ by $P_D,P_A$ and $P_S$, respectively. Choose any $z\in\mbb R$. Then 
    \begin{align*}
        \tilde F(z)
        &=P[(-\infty,z]\vert\mc E]\\
        &=\frac{P[(-\infty,z]\cap\mc E]}{P[\mc E]}\\
        &=\frac{p_DP_D[(-\infty,z]\cap\mc E]+p_AP_A[(-\infty,z]\cap\mc E]+p_SP_S[(-\infty,z]\cap\mc E]}{P[\mc E]}\\
        &\overset{(a)}{=}\frac{p_AP_A[(-\infty,z]]+p_SP_S[(-\infty,z]]}{P[\mc E]}\\
        &=F'(z),
    \end{align*}
    where $(a)$ is due to the specific properties of $P_D,P_A$ and $P_S$: We have $P_D[\mc E]=0$ because $P_D[\overline{\mc E}]=1$. Moreover, $P_A[\overline{\mc E}]=P_S[\overline{\mc E}]=0$ because $P_A$ and $P_S$ are continuous and $\overline{\mc E}$ is at most countably infinite \cite[Theorem 3 on p.~316]{K_Formin}, so $P_A[(-\infty,z]\cap\mc E]=P_A[(-\infty,z]]$ and $P_S[(-\infty,z]\cap\mc E]=P_S[(-\infty,z]]$.
\end{proof}

Let $L$ be any positive integer. By the continuity of $F'$ and the intermediate value theorem, there exist not necessarily unique numbers $z^*_1,\ldots,z^*_{L-1}$ satisfying
\[
    F'(z^*_l)=\frac{l}{L}.
\]
Hence if $F$ is continuous, then $F=F'$ and each of the intervals $(-\infty,z^*_1]$, $(z^*_1,z^*_2]$, \ldots, $(z^*_{L-2},z^*_{L-1}]$, $(z^*_{L-1},\infty)$ has probability $1/L$. For the general case, we have the following result.

\begin{lemma}\label{lem:CR_construction}
    Let $0<\eta<1$ and $L$ a positive integer. Let $Z_1,\ldots,Z_n$ be $n$ i.i.d.\ copies of the noise random variable $Z$ with non-discrete probability distribution $P$. Then for sufficiently large $n$, there exists a mapping $\pi:\mbb R^n\to\{1,\ldots,L\}\cup\{\infty\}$ such that for any $1\leq l\leq L$,
    \begin{equation}\label{eq:cond_unif}
        \Pr[\pi(Z_1,\ldots,Z_n)=l\vert \pi(Z_1,\ldots,Z_n)\neq\infty]=\frac{1}{L}
    \end{equation}
    and
    \begin{equation}\label{eq:CR_err_small}
        \Pr[\pi(Z_1,\ldots, Z_n)=\infty]\leq\eta.
    \end{equation}
\end{lemma}

\begin{proof}
    Choose $n$ so large that $P[\overline{\mc E}]^n\leq\eta$ (recall that $P[\overline{\mc E}]<1$ since $P$ is non-discrete). If $(z_1,\ldots,z_n)\in\overline{\mc E}^n$, then set $\pi(z_1,\ldots,z_n)=\infty$. Otherwise, let $j$ be the smallest index for which $z_j\in\mc E$ and set
    \[
        \pi(z_1,\ldots,z_n)=l\quad\text{if }z^*_{l-1}<z_j\leq z^*_l,
    \]
    where $z^*_0=-\infty$ and $z^*_L=+\infty$.
    
    It remains to show \eqref{eq:cond_unif} and \eqref{eq:CR_err_small}. For the latter relation, note that by the independence of $Z_1,\ldots,Z_n$,
    \begin{align*}
        &\Pr[\pi(Z_1,\ldots,Z_n)=\infty]
        =\Pr[Z_1\in\overline{\mc E},\ldots,Z_n\in\overline{\mc E}]=\Pr[Z\in\overline{\mc E}]^n
        =P[\overline{\mc E}]^n
        \leq\eta
    \end{align*}
    by choice of $n$. 
    
    In order to show \eqref{eq:cond_unif} we again use the independence of $Z_1,\ldots,Z_n$ and obtain
    \begin{align*}
        &\Pr[\pi(Z_1,\ldots,Z_n)=l\vert\pi(Z_1,\ldots,Z_n)\neq\infty]\\
        &=\sum_{j=1}^n\Pr\Bigl[Z_1\in\overline{\mc E},\ldots,Z_{j-1}\in\overline{\mc E},Z_j\in\mc E,z^*_{l-1}<Z_j\leq z^*_l\left\vert (Z_1,\ldots,Z_n)\notin\overline{\mc E}^n\right.\Bigr]\\
        &=\frac{1}{1-P[\overline{\mc E}]^n}\sum_{j=1}^n\Pr[Z_1\in\overline{\mc E},\ldots,Z_{j-1}\in\overline{\mc E},Z_j\in\mc E,z^*_{l-1}<Z_j\leq z^*_l]\\
        &=\frac{1}{1-P[\overline{\mc E}]^n}\sum_{j=1}^n\Pr[Z\in\overline{\mc E}]^{j-1}\Pr[z^*_{l-1}<Z\leq z^*_l\vert Z\in\mc E]\Pr[Z\in\mc E]\\
        &=\frac{1}{L(1-P[\overline{\mc E}]^n)}\sum_{j=1}^nP[\overline{\mc E}]^{j-1}P[\mc E]\\
        &=\frac{1}{L}.
    \end{align*}
    This completes the proof.
\end{proof}

We will use the function $\pi$ constructed in the lemma to establish a uniform random experiment on the set $\{1,\ldots,L\}$ shared by sender and receiver. This does not succeed with certainty, but note that the failure probability, which is at most $\eta$, is independent of $L$.

\subsection{Proof of Theorem \ref{theorem:new_main_theorem}} \label{subsection:codingscheme}

In this section, we provide a proof for Theorem \ref{theorem:new_main_theorem}. We first give the proof for the case of an average power constraint. At the end of the subsection, we will show how to modify the code in the case of a peak power constraint; the required changes are minimal and straightforward.

The feedback code we construct consists of two parts. In the first one, we use the channel $n$ times to establish uniform common randomness as described in the previous subsection. The second part, of the same blocklength $n$, is used for the transmission of the color value according to which a receiver determines whether its identity was sent, as described already in the sketch of our strategy at the beginning of this section. The total blocklength of our coding scheme will thus be $2n$. Some optimization would certainly be possible, but we concentrate here on showing the possibility of distinguishing an arbitrary number of identities.

\subsubsection{Code construction}

Let $0<\lambda<1/2$ and $\Gamma>0$. For sufficiently large $n$ and any positive integer $N$, we are going to construct a $(2n,N,\lambda,\lambda)$ identification feedback code with average power constraint $\Gamma$ consisting of two parts of length $n$ each. As a first condition on $n$, we impose that
\begin{equation}\label{eq:n_cond}
    P[\overline{\mc E}]^n\leq\frac{\lambda}{2}.
\end{equation}
There are two main ingredients for our construction. First of all, we need an $(n,M,\lambda/2)$ message transmission feedback code $\{(u_m,\mc D_m'):1\leq m\leq M\}$ satisfying the average power constraint $\Gamma$ and
\begin{equation}\label{eq:M_and_lambda}
    M>\frac{2}{\lambda}.
\end{equation}
Such a code exists for sufficiently large $n$ by the assumption that $W_P$ has a positive message transmission feedback capacity with average power constraint $\Gamma$.

The second ingredient to our construction is a family of coloring functions $\{k_i:1\leq i\leq N\}$, one for each identity, such that $k_i:\{1,\ldots,L\}\to\{1,\ldots,M\}$ for all $i$. The properties we need $L$ and the $k_i$ to satisfy will be given below.

\begin{figure*}
\centering
\scalebox{1}{\begin{tikzpicture}
		\node[rectangle split, rectangle split parts=6,draw,inner sep=1ex] (A) at (0,1)
		{$1$\nodepart{two}$2$ \nodepart{three}$\vdots$\nodepart{four}$l$\nodepart{five}$\vdots$\nodepart{six}$L$ };
		
		\node[rectangle split, rectangle split parts=6,draw,inner sep=.5ex] (B) at (2,1)
		{$1$\nodepart{two}$2$ \nodepart{three}$ \vdots$\nodepart{four}$j$\nodepart{five}$ \vdots$\nodepart{six}$M$ };
		\draw[->] (A.text east) -- (B.two west);
		\draw[->] (A.two east) -- (B.two west);
		\draw[->] (A.three east) -- (B.four west);
		\draw[->] (A.five east) -- (B.four west);
		\draw[->] (A.four east) -- (B.six west);
		\draw[->] (A.six east) -- (B.six west);
		\node (C) at (1,-1) {$k_1$};
		
		\node[rectangle split, rectangle split parts=6,draw,inner sep=1ex] (A1) at (8,1)
		{$1$\nodepart{two}$2$ \nodepart{three}$\vdots$\nodepart{four}$l$\nodepart{five}$\vdots$\nodepart{six}$L$ };
		
		\node[rectangle split, rectangle split parts=6,draw,inner sep=.5ex] (B1) at (10,1)
		{$1$\nodepart{two}$2$ \nodepart{three}$ \vdots$\nodepart{four}$j$\nodepart{five}$ \vdots$\nodepart{six}$M$ };
		
		\draw[->] (A1.text east) -- (B1.two west);
		\draw[->] (A1.four east) -- (B1.two west);
		\draw[->] (A1.three east) -- (B1.four west);
		\draw[->] (A1.two east) -- (B1.four west);
		\draw[->] (A1.five east) -- (B1.six west);
		\draw[->] (A1.six east) -- (B1.six west);
		\node (C1) at (9,-1) {$k_N$};
		
		\node (D) at (5,1) {$\ldots$};
		\node (E) at (5.5,1) {$\ldots$};
		
		%\node (T) at (5,-1.1) {\color{red}{$F_1(1)=F_N(1)=2$}};
		\end{tikzpicture}}
    \caption{Typical coloring functions}
    \label{fig:typical_coloring_functions}
\end{figure*}
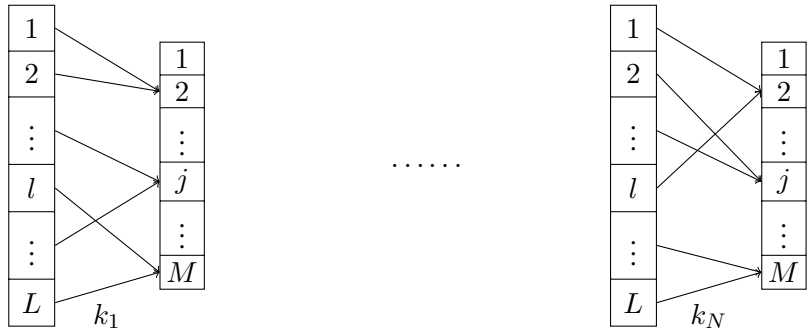

We first show how to obtain our identification feedback code from these ingredients. They are connected by the function $\pi:\mbb R^n\to\{1,\ldots,L\}$ constructed in Lemma \ref{lem:CR_construction}. For identity $i\in\{1,\ldots,N\}$, we define the encoding feedback function $f_i$ as follows. We set $f_i^1=\ldots=f_i^n=0$, so no feedback is used in the first $n$ steps. For $n+1\leq t\leq 2n$, we set
\begin{align*}
    f_i^t(y_1,\ldots,y_{t-1})
    =
    \begin{cases}
        0 & \text{if }\pi(y_1,\ldots,y_n)=\infty,\\
        u_{k_i(l)}^t(y_{n+1},\ldots,y_{t-1}) & \text{if }\pi(y_1,\ldots,y_n)=l\in\{1,\ldots,L\}.
    \end{cases}
\end{align*}
In other words, if the generation of common randomness between the sender and the receiver succeeded, the components $n+1,\ldots,2n$ of $f_i$ are formed by the encoding feedback function $u_{k_i(l)}$ of the message transmission code which transmits the message $k_i(l)$ resulting from the random experiment performed in the first $n$ channel uses. Since we can use the symbol 0 in the first $n$ channel uses, every $f_i$ satisfies the average power constraint $\Gamma$ because every $u_m$ does.

We define the decoding sets $\mc D_i$ indirectly by defining a decision function $\psi_i:\mbb R^{2n}\to\{0,1\}$ for every identity $i$. We then set $\mc D_i=\psi^{-1}(1)$, so the receiver with identity $i$ decides that identity $i$ was sent if $\psi_i(y_1,\ldots,y_{2n})=1$. The condition for $\psi_i(y_1,\ldots,y_{2n})=1$ is that
\begin{align*}
    &\pi(y_1,\ldots,y_n)=l\qquad\text{ for some }l\in\{1,\ldots,L\}\text{ and }
    (y_{n+1},\ldots,y_{2n})\in\mc D'_{k_i(l)}.
\end{align*}
In all other cases, we set $\psi_i(y_1,\ldots,y_{2n})=0$. As promised, if the common randomness generation between sender and receiver succeeds, the receiver checks whether the message $\hat m$ transmitted in steps $n+1,\ldots,2n$ matches the color $k_i(l)$, and decides that $i$ was sent if and only if this is the case.

\subsubsection{The coloring functions}

For the coloring functions, we require that they have a large pairwise Hamming distance. Formally, for any two distinct identities $i$ and $j$, we want
\begin{equation}\label{eq:large_Hamming}
    \lvert\{l:k_i(l)=k_j(l)\}\rvert\leq\frac{\lambda L}{2}.
\end{equation}
This condition ensures that the overlaps of the decoding sets $\mc D_i$ are not too large, which is important for upper-bounding the error probabilities of the second kind. We are going to choose the $k_i$ at random in order to show that it is possible to find $N$ coloring functions satisfying \eqref{eq:large_Hamming} if $L$ is sufficiently large. We follow the same method as \cite{IDFeedback}. Typical coloring functions look as shown in Fig.~\ref{fig:typical_coloring_functions}.

Let $\{K_i(l):1\leq i\leq N,1\leq l\leq L\}$ be i.i.d.\ random variables uniformly distributed on $\{1,\ldots,M\}$. The values $K_i(1),\ldots,K_i(L)$ are the values of the random coloring function $K_i$. We bound the probability that the Hamming distance of the random coloring functions of any two distinct identities $i$ and $j$ is too large. For this purpose, define the random variables $H_1^{(i,j)},\ldots,H_L^{(i,j)}$ by 
\[
    H_l^{(i,j)}=
    \begin{cases}
        1 & \text{if }K_i(l)=K_j(l),\\
        0 & \text{else}.
    \end{cases}
\] 
We will need Hoeffding's inequality.

\begin{lemma}[Hoeffding's inequality, Theorem 1 in \cite{HoeffdingArticle}]  \label{HoeffIneq}
 Let $\{X_i\}$ be i.i.d.\ random variables taking values in $[0,1]$ with mean $\mu$. Then for all $c >0$ with $\mu+c \leq 1$, 
 \begin{equation*}
 \Pr\left[\frac{1}{n} \sum_{i=1}^{n} X_i- \mu \geq c \right]\leq 2^{-\frac{2nc^2}{\ln 2}}.
 \end{equation*}
\end{lemma}

\begin{lemma} \label{lemmahof}
For $\lambda \in (0,1/2)$ satisfying \eqref{eq:M_and_lambda} and any pair $(i,j)$ of distinct identities in $\{1,\ldots,N\}$,
\begin{align*}
    \Pr\left[\sum_{l=1}^L H_l^{(i,j)} > \frac{\lambda L}{2} \right]  \leq 2^{-\frac{L\lambda c^2}{\ln 2}},
\end{align*} 
where $c=\frac{\lambda}{2}-\frac{1}{M}>0$.
\end{lemma}

\begin{proof}
    Without loss of generality, we set $i=1$ and $j=2$ and write $H_l$ instead of $H_l^{(1,2)}$. We will show the inequality stated in the lemma conditional on any fixed realization of $K_1$. Then
    \begin{align*}
        \Pr\left[\sum_{l=1}^L H_l > \frac{\lambda L}{2} \right]
        =\sum_{k_1}\Pr\left[\left.\sum_{l=1}^L H_l > \frac{\lambda L}{2}\right\vert K_1=k_1 \right]\Pr[K_1=k_1]
        \leq 2^{-\frac{L\lambda c^2}{\ln 2}},
    \end{align*}
    where the sum is over all possible realizations $k_1$ of $K_1$.
    
    Conditional on the event $K_1=k_1$, the random variables $H_1,\ldots,H_L$ are i.i.d.\ and satisfy $\mbb E[H_l\vert K_1=k_1]=\Pr[H_l=1\vert K_1=k_1]=1/M$. Thus we can apply Hoeffding's inequality and obtain
    \[
        \Pr\left[\left.\sum_{l=1}^L H_l > \frac{\lambda L}{2}\right\vert K_1=k_1 \right]
        \leq 2^{-\frac{L\lambda c^2}{\ln 2} },
    \]
    as claimed.
\end{proof}

It remains to choose the $K_i$ in such a way that \eqref{eq:large_Hamming} is satisfied for all pairs of distinct identities simultaneously. 

\begin{lemma}\label{lem:ex_good_fam}
    If 
    \begin{equation}\label{eq:N_L_condition}
        \frac{\lambda L c^2}{2\ln 2}>\log N,
    \end{equation}
    then there exists a family $k_1,\ldots,k_N$ of coloring functions satisfying \eqref{eq:large_Hamming} for all pairs $(i,j)$ of distinct identities.
\end{lemma}

\begin{proof}
    By the union bound and Lemma \ref{lemmahof}, we have
    \begin{align*}
        \Pr\left[\sum_{l=1}^L H_l^{(i,j)} > \frac{\lambda L}{2}\text{ for some }i\neq j \right]
        \leq N(N-1)\Pr\left[\sum_{l=1}^L H_l > \frac{\lambda L}{2} \right]
        \leq N(N-1)2^{-\frac{L\lambda c^2}{\ln 2} }.
    \end{align*}
    This probability is strictly smaller than 1 if \eqref{eq:N_L_condition} is satisfied. Hence there exists a realization $k_1,\ldots,k_N$ of the random coloring functions $K_1,\ldots,K_N$ such that \eqref{eq:large_Hamming} is satisfied for all distinct identities $i,j$.
\end{proof}

We now choose any $L$ satisfying \eqref{eq:N_L_condition}. Lemma \ref{lem:ex_good_fam} ensures the existence of a family $k_1,\ldots,k_N$ satisfying \eqref{eq:large_Hamming} for all distinct identities $i,j$. We use this family of coloring functions in our code construction described above.

\subsubsection{Error analysis}

We now bound the error probabilities $\mu_1^{(i)}$ and $\mu_2^{(i,j)}$ for our identification feedback code. We start with $\mu_1^{(i)}$. In order to simplify notation, we introduce the random variable $\Pi=\pi(Y_1,\ldots,Y_n)$. By the definition of the decoding set $\mc D_i$, 
\begin{align*}
    \mu_1^{(i)}
    &=W^{2n}_{P}({\setd}_i^c|f_i)\\
    &\overset{(a)}{\leq} \Pr\left\{\Pi=\infty \right\}+\mathbb{E}\left[\left. W^n_{P} \left( ( \mc D'_{k_i(\Pi)})^c | u_{k_i(\Pi)} \right) \right\vert \Pi\neq\infty \right] \\
    &\overset{(b)}{\leq} P[\overline{\mc E}]^n + \frac{\lambda}{2}\\
    &\overset{(c)}{\leq}\lambda,
\end{align*}
where $(a)$ follows from the memorylessness property of the channel, $(b)$ follows from the properties of the message transmission feedback code and of $\Pi$ and $(c)$ is due to the choice of $n$ in \eqref{eq:n_cond}.
    
Now let $i$ and $j$ be two distinct identities. Then
\begin{align*}
    &\mu_2^{(i,j)}\\
    &=W^{2n}_{P}(\mc D_j|f_i) \\
    &\leq\mbb E\left[\left.W_P^n\left(\mc D'_{k_j(\Pi)}\vert u_{k_i(\Pi)}\right)\right\vert\Pi\neq\infty\right]\\
    &\leq\sum_{l=1}^L\mbb E\left[\left.W_P\left(\mc D'_{k_j(l)}\vert u_{k_i(l)}\right)\right\vert\Pi=l\right]\Pr[\Pi=l\vert\Pi\neq\infty]\\
    & \overset{(a)}{=}\frac{1}{L}\sum_{l:k_i(l)\neq k_j(l)} \mathbb{E}\left[\left. W^n_{P} \left( \mc D'_{k_j(l)} | u_{k_i(l)} \right) \right\vert \Pi=l \right]+\frac{1}{L}\sum_{l:k_i(l)=k_j(l)} \mathbb{E}\left[\left. W^n_{P} \left( \mc D'_{k_j(l)}| u_{k_i(l)} \right) \right\vert \Pi=l\right] \nonumber\\
    & \overset{(b)}{\leq} \frac{\lambda}{2}+\frac{\lambda}{2}\\
    &=\lambda,
\end{align*}
where $(a)$ follows from the memorylessness property of the channel and the properties of $\pi$ and $(b)$ uses the upper bound on the error probability of the message transmission feedback code in the first summand and property \eqref{eq:large_Hamming} of the family of coloring functions in the second summand.

We have constructed an $(2n,N,\lambda,\lambda)$ identification feedback code satisfying the average power constraint $\Gamma$, where $n$ only depends on $\lambda$ and $N$ can be chosen arbitrarily large. In order to obtain an identification feedback code satisfying the peak power constraint $\Gamma$, all one has to do is to replace the message transmission feedback code used above by a message transmission feedback code satisfying the peak power constraint $\Gamma$ as well as \eqref{eq:M_and_lambda}.

\begin{remark}
    \begin{enumerate}
        \item Obviously, we are not required to employ message transmission feedback codes which really use the feedback. Message transmission codes without feedback are sufficient if the message transmission capacity of $W_P$ without feedback is positive, as in the case of the AWGN channel. Nevertheless, the identification feedback codes we have used in the proof of Theorem \ref{theorem:new_main_theorem} will still be feedback strategies.
        \item Ahlswede and Dueck \cite{IDFeedback} use the same construction in the achievability proof for their coding theorem on deterministic identification over DMCs with feedback. In order to achieve the maximal possible rate, they have to ensure that the length of the message transmission code is negligible with respect to the number of channel uses needed for common randomness generation. For this reason, they take $n$ channel uses for common randomness generation and $\sqrt n$ for the transmission of the colors. In our case, this matter is much less sensitive since our maximal rate is infinite.
    \end{enumerate}
\end{remark}

\subsection{Proof of Corollary \ref{corrolary}}

Not much work has to be done in order to prove the corollary. Again we only treat the case where we have an average power constraint $\Gamma$. Let $\varphi$ be an arbitrary rate function and $R$ a nonnegative real number. We will show that $R$ is an achievable rate w.r.t.\ $\varphi$ and with average power constraint $\Gamma$.

Let $0<\lambda<1/2$. By Theorem \ref{theorem:new_main_theorem}, for sufficiently large $n$ and any number $N$ of identities we can find an $(n,N,\lambda,\lambda)$ identification feedback code satisfying the average power constraint $\Gamma$. In particular, such a code exists for any $N\geq\varphi(nR)$. The proof is complete.

\section{Conclusions} \label{sec: conclusions}

In this paper we considered message identification via channels with non-discrete additive white noise in the presence of noiseless feedback and without local randomness. We showed that if the channel has a positive message transmission feedback capacity, for given error thresholds and sufficiently large blocklength we can construct arbitrarily large deterministic identification codes. This holds for both average and peak power constraints. This is a highly surprising result and shows that the addition of perfect feedback can result in a tremendous capacity gain compared with the non-feedback setting. For instance, it is known that the deterministic identification capacity in the case of Gaussian additive white noise without feedback is finite for the scaling functions $\varphi_2$ defined above \cite{deterministicDMC}. Even when allowing randomized encoding, the capacity still is finite without feedback, although for a different scaling function, $\varphi_3$. 

Moreover, for traditional message transmission, feedback does not increase the capacity of the channel. Thus, the identification problem here shows a completely new phenomenon: By adding noiseless feedback to the channel, the capacity becomes infinitely large, regardless of the scaling. On closer inspection, it is due to the fact that feedback allows the generation of infinite common randomness between the sender and the receiver. This common randomness can be used to make the ambiguity in the receiver's decisions arbitrarily small no matter how many identities need to be distinguished.

Motivated by the drastic effects produced by the common randomness obtained from the perfect feedback in the model treated in this work, it would be interesting to investigate common randomness generation from continuous correlated sources in the future. Another related problem not treated so far is identification in the presence of noisy feedback over continuous channels. In this case, we do not expect the effect of scale-independent infinite capacity observed for noiseless feedback to occur.

\section*{Acknowledgments}

M.~Wiese was supported by the German Research Foundation (DFG) within Germany’s Excellence Strategy EXC-2092 CASA-390781972 and within the Gottfried Wilhelm Leibniz Prize under Grant BO 1734/20-1. W.~Labidi was supported by the German Federal Ministry of Education and Research (BMBF) within the national initiative for “Post Shannon Communication (NewCom)” through the project “Basics, simulation and demonstration for new communication models” under Grant 16KIS1003K and within the national initiative for ``Molecular Communications'' (MAMOKO) under Grant 16KIS0914. C.~Deppe was supported by BMBF within NewCom through the project “Coding theory and coding methods for new communication models” under Grant 16KIS1005 and within the national initiative on 6G Communication Systems through the research hub 6G-life under Grant 16KISK002. H.~Boche was supported in part by BMBF within 6G-life under Grant 16KISK002, within NewCom under Grant 16KIS1003K, by the Bavarian Ministry of Economic Affairs, Regional Development and Energy as part of the project 6G Future Lab Bavaria, and by DFG within the Gottfried Wilhelm Leibniz Prize under Grant BO 1734/20-1.

\bibliographystyle{IEEEtranS}
\bibliography{../definitions,../references}

\end{document}